\documentclass[11pt,a4paper,english]{article}

\usepackage[bindingoffset=0.3cm,textheight=22.5cm,hdivide={2.7cm,*,2.7cm}, vdivide={*,22cm,*}]{geometry}
\usepackage[dvips,bookmarksnumbered=true,breaklinks=true]{hyperref}

\usepackage{amsmath,amsfonts,amssymb,babel}

\usepackage[numbers,square,comma,sort&compress]{natbib}

\IfFileExists{dsfont.sty}
	{\usepackage{dsfont}
         \newcommand{\id}{\mathds{1}}}
	{\typeout{Package dsfont.sty was not found, using alternative macros.}
         \let\mathds=\mathbb
         \newcommand{\id}{\mbox{1 \kern-.59em {\rm l}}}}
\let\one=\id

\makeatletter

         \@addtoreset{equation}{section}
\makeatother

\newcommand{\nocontentsline}[3]{}
\newcommand{\tocless}[3]{\bgroup\let\addcontentsline=\nocontentsline#1{#2}#3\egroup}

\newtheorem{theorem}{Theorem}
\newtheorem{lemma}[theorem]{Lemma}

\newcommand{\qed}{\nobreak \ifvmode \relax \else
      \ifdim\lastskip<1.5em \hskip-\lastskip
      \hskip1.5em plus0em minus0.5em \fi \nobreak
      \vrule height0.75em width0.5em depth0.25em\fi}

\newenvironment{proof}[1][Proof. \hspace*{1ex}]{\begin{trivlist}
\item[\hskip \labelsep {\bfseries #1}]}{\qed\end{trivlist}}

\newcommand{\be}{\begin{equation}}
\newcommand{\ee}{\end{equation}}
\newcommand{\eq}[1]{(\ref{#1})}

\def\nn{\nonumber}
\def\bea{\begin{eqnarray}}
\def\eea{\end{eqnarray}}

\def\beqa{\begin{eqnarray}} 
\def\eeqa{\end{eqnarray}} 
\def\beq{\begin{equation}} 
\def\eeq{\end{equation}}

\def\Tr{{\rm Tr}}

\def\a{\alpha}          

\def\d{\delta}

 \def\L{\Lambda} 
\def\m{\mu}     \def\n{\nu}

\def\r{\rho}
\def\s{\sigma}  

\def\th{\theta}

 \def\cH{{\cal H}} 
  
\def\cM{{\cal M}} \def\cN{{\cal N}}

\newcommand{\R}{\mathds{R}}
\newcommand{\C}{\mathds{C}}

\newcommand{\Z}{\mathds{Z}}

\def\bit{\begin{itemize}}
\def\eit{\end{itemize}}

\def\({\left(}
\def\){\right)}
\def\diag{\mbox{diag}}

\def\Mat{{\rm Mat}}

\def\d{\delta}

\def\pa{\partial} \def\del{\partial}

\def\bcomment#1{}

\newcommand{\co}[2]{[#1,#2]}						
\renewcommand{\a}{\alpha}

\renewcommand{\d}{\delta}

\renewcommand{\th}{\theta}

\renewcommand{\r}{\rho}

\renewcommand{\L}{\Lambda}
\renewcommand{\Xi}{\Xi}

 \allowdisplaybreaks[3]

\begin{document}

\renewcommand{\title}[1]{\vspace{10mm}\noindent{\Large{\bf
#1}}\vspace{8mm}} \newcommand{\authors}[1]{\noindent{\large
#1}\vspace{5mm}} \newcommand{\address}[1]{{\itshape #1\vspace{2mm}}}


\begin{flushright}
UWThPh-2011-20
\end{flushright}

\begin{center}

\title{ \Large  Split noncommutativity and compactified brane solutions \\[1ex] in matrix models}

\vskip 3mm

\authors{Harold {Steinacker{\footnote{harold.steinacker@univie.ac.at}}}
}

\vskip 3mm

\address{ {\it Faculty of Physics, University of Vienna\\
 Boltzmanngasse 5, A-1090 Vienna (Austria)  }}

\vskip 1.4cm

\textbf{Abstract}

\end{center}

Solutions of the undeformed IKKT matrix model with structure $\R^{3,1} \times K$ are presented,
where the noncommutativity relates the compact with the non-compact space. 
The extra dimensions are stabilized by angular momentum, and the
scales of $K$ are generic moduli of the solutions. 
Explicit solutions are given for $K= T^2, K = S^3 \times S^1, \,  K = S^2 \times T^2$ and $K = S^2 \times S^2$.
Infinite towers of Kaluza-Klein modes may arise in some directions, 
along with an effective UV cutoff on the non-compact space.
Deformations of these solutions carry NC gauge theory 
coupled to (emergent) gravity. 
Analogous solutions of the BFSS model are also given.

\vskip 1.4cm

\tableofcontents


%
%
%
%
%
%
%
%
%
%

\section{Introduction}\label{sec:background}

Matrix models such as the IKKT respectively IIB model \cite{Ishibashi:1996xs}
provide fascinating candidates for a quantum theory of 
fundamental interactions. 
Part of the appeal stems from the fact that geometry is not an input, but emerges on suitable solutions. 
For example, it is easy to see that flat noncommutative (NC) planes $\R^{2n}_\theta$ are solutions.
More generally, one can consider geometric deformations of such brane solutions which correspond
to embedded NC branes $\cM \subset \R^{10}$. Their effective geometry can be described easily
in the semi-classical limit \cite{Steinacker:2008ri,Steinacker:2010rh}, 
resulting in a dynamical effective metric which is strongly reminiscent of the
open string metric  \cite{Seiberg:1999vs}. Since the maximally supersymmetric matrix model 
is expected to provide a good quantum theory on 4-dimensional branes, 
a physically interesting quantum theory of gravity should arise on such branes.

However, a single 4-dimensional brane $\cM^4 \subset \R^{10}$ is clearly too simple to reproduce the rich 
spectrum of phenomena in nature. In order to recover e.g. the standard model, additional structure is needed.
One possible origin of such additional structure are compactified extra dimensions, as considered 
in string theory. Another very interesting possibility are intersecting branes, which play an essential
role in recent attempts to recover the standard model from string theory. 
At present, it appears that intersections of branes with compactified extra dimensions provide the most promising 
avenue towards realistic physics, cf. \cite{Blumenhagen:2006ci}.

It is well-known that compactified extra dimensions such as fuzzy spheres $S^2_N$ arise as solutions 
of matrix models {\em with additional terms}, such as quadratic of cubic terms 
\cite{Iso:2001mg,Myers:1999ps,Alekseev:2000fd,Kimura:2001uk,Aschieri:2006uw}. This should allow
in principle to obtain sufficiently rich solutions in order to recover  the standard model.
However, the addition of such extra terms spoils much of the appeal of the matrix model: the geometry
of the space-time branes is then strongly constrained, 
and much of the essential (super)symmetries is lost.

In the present paper, we show that there are indeed solutions of the {\em undeformed} IKKT 
 model with geometry $\R^4 \times K$, where $K$ can be $S^2$, $T^2$, or
$S^3\times S^1, S^2 \times S^2,$ and $S^2 \times T^2$.
For all these solutions, the Minkowski signature of the model is essential.
Analogous solutions for the BFSS model are quite obvious
and known to some extent; however,  these solutions of the IKKT model
appear to be new. Our constructions are inspired by the 
well-known ``supertube''  or fuzzy cylinder solution 
\cite{Bak:2001kq,Janssen:2004jz,Terashima:2007uf,Shepard:2005wy,Bak:2002wy}
of the BFSS model \cite{Banks:1996vh,de Wit:1988ig}, carried over by a 
twisting procedure to the IKKT case.
These solutions should provide sufficient structure towards physically realistic solutions 
of the matrix model, along the lines of \cite{intersect-paper}.

All the solutions under consideration here have an interesting common feature: the 
noncommutative structure $\theta^{ab}$ -- which underlies all interesting matrix model solutions -- 
does not respect the compact resp. non-compact spaces in $\R^4 \times K$, but connects them in an essential
way. The structure is reminiscent of the canonical symplectic structure of cotangent bundles, 
where compact and non-compact coordinates are canonically conjugated.
This is indicated by the name   ``split noncommutativity''.

There are many  reasons why split noncommutativity is interesting. 
First, the solutions presented here only exist in the case of Minkowski signature. This is of course 
welcome from a physical point of view\footnote{This also suggests that the Euclidean model may not be suitable 
to understand the vacuum structure.}.
The underlying mechanism is that the compact extra dimensions are stabilized by (internal) angular momentum.
Moreover, there are even solutions whose
non-compact sector is in fact commutative and hence isotropic, along with an intrinsic UV cutoff. 
This would seem to resolve many of the problems associated with noncommutative field theory - 
the violation of Lorentz invariance, causality, etc. - which should be  hidden in the  
compact sector. 

However, things are not that simple. In the Minkowski case, the effective (open string) metric $G_{\mu\nu}$
has a different causality structure\footnote{As discussed in \cite{Steinacker:2010rh}, 
this change of causality structure might be avoided using 
complexified $\theta^{\mu\nu}$. A similar issue would also arise in string theory with
time-like fluxes. The appropriate treatment of this issue is unclear.}  than the naive 
embedding  (closed string)  metric $g_{\mu\nu}$. Thus the completely isotropic solutions turn out to be non-propagating
with compact time-like curves, which is clearly undesirable. 
On the other hand, 
we do obtain physically meaningful solutions with standard Minkowski metric on the non-compact
space-time, at the expense of admitting some space-time noncommutativity in the non-compact sector.

The geometrical degrees of freedom provided by the extra dimensions 
are also  welcome for the effective (emergent) gravity on such branes. 
The higher-dimensional compactification provides additional degrees of freedom
associated to gravity  due to the Poisson structure \cite{Rivelles:2002ez,Steinacker:2007dq,Yang:2006mn}, 
and fewer embedding degrees of freedom. 
Moreover, the scale and to some extent the shape
of the compact space are free moduli of the solutions, and not fixed by the model.
We will discuss these gravitational aspects only briefly in this paper, and postpone 
a systematic analysis to future work.

\section{Matrix models and their geometry} 
\label{sec:matrixmodels-intro}

We briefly collect the essential ingredients of the matrix model framework
and its effective geometry, referring 
to the recent review \cite{Steinacker:2010rh} for more details.

\subsection{The IKKT matrix model}
\label{sec:basic}

The starting point is given by a matrix model of Yang-Mills type, 
\begin{align}
S_{YM}&=-\Tr\co{X^a}{X^b}\co{X^c}{X^d}\eta_{ac}\eta_{bd}\,,
\label{S-YM}
\end{align}
where the indices run from $0$ to $9$, and $\eta_{ac} = (-1,1,...,1)$ is the invariant tensor of $SO(9,1)$. 
This is the bosonic sector of the  the 10-dimensional
maximally supersymmetric IKKT or IIB model \cite{Ishibashi:1996xs}.
The ``covariant coordinates'' 
$X^a$ are Hermitian matrices, i.e. operators acting on a separable Hilbert space $\mathcal{H}$. 
The equations of motion take the following simple form
\be
\,[X_b,[X^b,X^a]] = 0
\label{eom-IKKT}
\ee
for all $a$. Indices of matrices will be raised or lowered with $\eta_{ab}$. 
We denote the commutator of two matrices as
\begin{align}
\co{X^a}{X^b}&=i \Theta^{ab}\,.
\end{align}
We focus on matrix configurations which describe embedded
noncommutative (NC) branes. 
This means that the $X^a$ are quantized embedding functions
\be
X^a \sim x^a: \quad \mathcal{M}^{2n}\hookrightarrow \R^{10} 
\ee 
of a $2n$ dimensional submanifold, and
\begin{align}
\co{X^\m}{X^\n}&\sim i \th^{\m\n}(x)\,
\end{align}
is interpreted as quantized a Poisson structure on $\mathcal{M}^{2n}$.
Here $\sim$ denotes the 
semi-classical limit where commutators are replaced by Poisson brackets, and 
$x^\mu$ are locally independent coordinate functions chosen among the $x^a$.
Such a collection of matrices defines a quantized embedded
Poisson manifolds $(\cM^{2n},\theta^{\mu\nu})$, denoted as ``matrix geometry''. 
We will assume that $\th^{\m\n}$ is non-degenerate, so that its 
inverse matrix $\th^{-1}_{\m\n}$ defines a symplectic form 
on $\mathcal{M}^{2n}$.
The sub-manifold $\cM^{2n}\subset\R^{10}$ 
is equipped with a non-trivial induced metric
\begin{align}
g_{\m\n}(x)=\pa_\m x^a \pa_\n x^b\eta_{ab}\,,
\label{eq:def-induced-metric}
\end{align}
via the pull-back of $\eta_{ab}$. 
Finally, we define the following (effective) metric 
\begin{align}
G^{\m\n}&= \th^{\m\r}\th^{\n\s}g_{\r\s}\,, 
\label{eff-metric}
\end{align}
on $\cM^{2n}$,
dropping possible conformal factors  \cite{Steinacker:2008ri}
which are not of interest here.
It is not hard to see that the kinetic term for scalar fields on  $\cM^{2n}$
is governed by the effective metric $G_{\mu\nu}(x)$. 
The same metric also governs non-Abelian gauge fields and fermions  
on $\cM$ (up to possible conformal factors) \cite{Steinacker:2008ri,Steinacker:2008ya}, 
so that $G_{\mu\nu}$ {\em must} be interpreted as
gravitational metric. 
Since the embedding is dynamical, the model describes a dynamical
theory of gravity, realized on dynamically 
determined submanifolds of $\R^{10}$.

\subsection{The BFSS matrix model}

Although our focus is on the  IKKT model, it is very instructive to recall
also the BFSS model, which was proposed as a non-perturbative definition of M(atrix) theory \cite{Banks:1996vh},
cf. \cite{de Wit:1988ig,hoppe}.
It is a time-dependent matrix model with 9 bosonic matrices 
$X^a(t)$, and appropriate fermions. 
Rather than discussing the action, we only write down here the bosonic equations of motion:
\be
\ddot X^a + [X^b , [X_b , X^a ]] = 0 , \qquad (a,b  = 1, . . . , 9) ,
\label{eom-BFSS}
\ee
dropping all dimensionful parameters.

\section{Extra dimensions and split noncommutativity}
\label{sec:split-NC}

The basic idea of this paper is to study NC brane configurations in the matrix model with geometry
$\cM^{2n} =\cM^4 \times K$, where the noncommutative structure  mixes the spacetime $\cM^4$ with the compact 
space $K$.
This means that the
(non-degenerate) Poisson structure $\Pi$ on  $\cM^{2n}$ satisfies
$\Pi(dx\wedge dy) \neq 0$, so that it
contains terms of the form
\be
\Pi = \theta^{\mu i}(x,y) \frac{\del}{\del x^\mu} \wedge \frac{\del}{\del y^i}\quad + ...
\label{split-NC}
\ee 
where $x^\mu$ are coordinates on $\cM^4$ and $y^i$ are coordinates on $K$. 
This will be indicated by the name ``split noncommutativity''.
If $\cM$ and $K$ have the same dimension, then we may even impose $\Pi(dx^\mu\wedge dx^\nu) =0$, i.e.
$\cM$ is isotropic. 
A standard example is the canonical symplectic structure on the cotangent bundle $T^* K$.

There are several reasons why split noncommutativity is interesting. First, if $\cM^4$ is 
an isotropic submanifold, then 
- as the name indicates -  it does not carry any Poisson tensor field which could break  Lorentz invariance.
Indeed there are very strong bounds on Lorentz violation, and thus on possible Poisson background fields.
Moreover, we will see that (maximally) split noncommutativity implies an
effective UV cutoff on $\cM$, due to the NC structure on $\cM^4 \times K$; 
this will be discussed in section \ref{sec:gauge}.
Another motivation is that such a structure will allow us to find solutions of Yang-Mills matrix
models with compact extra dimensions, {\em without} any additional terms in the action that would 
break some symmetry or introduce scale parameters\footnote{There are solutions of the IKKT model which can be 
interpreted as compactification on a torus \cite{Connes:1997cr}.
However these are very different types of ``stringy'' 10-dimensional solution involving infinite (winding) sectors, 
which do not fit into the framework of embedded branes under consideration here.
The IIB model is divergent on such solutions at one loop, but it  is expected to be finite on the present 
(lower-dimensional) solutions.}. 
In particular the shape and scale parameters are {\em free} moduli, which means that 
these solutions should admit deformations with nontrivial effective 4D geometry.
This should be important for (emergent) gravity, where some of the metric degrees of freedom come from 
the brane embedding. These aspects will be discussed briefly in 
sections \ref{sec:moduli} and \ref{sec:gravity}.

\subsection{Basic example: the fuzzy cylinder}

A simple prototype  of a space with split NC
is given by the fuzzy cylinder \cite{Janssen:2004jz,Chaichian:1998kp,Bak:2001kq}
$S^1 \times_\xi \R$: 
\begin{align}
[X^1,X^3] &= i \xi X^2, & [X^2,X^3] &= - i \xi X^1,  \nn\\
 (X^1)^2 + (X^2)^2 &= R^2, & [X^1,X^2] &= 0  .
\end{align}
Defining $U := X^1+iX^2$ and  $U^\dagger := X^1-iX^2$,
this can be stated more transparently as 
\bea
U U ^\dagger &=& U^\dagger U \quad = R^2     \nn\\   
\, [U, X^3] &=& \xi U,  \qquad [U^\dagger, X^3] = - \xi U^\dagger  
 \label{fuzzy-cylinder} 
\eea
This algebra has the following 
irreducible representation\footnote{More general irreducible representations are obtained from
this basic representation by a (trivial) constant shift $X^3 \to X^3 + c$.}
\begin{align}
U |n\rangle &= R   |n + 1\rangle , \qquad   U^\dagger  |n\rangle = R  |n - 1\rangle  \nn\\
X^3 |n\rangle &= \xi n  |n\rangle,  \qquad \qquad n \in \Z, \,\, \xi \in \R
\label{cylinder-rep}
\end{align}
on a Hilbert space $\cH$, where $|n\rangle$ form an orthonormal basis. 
We take $\xi \in \R$, since  the $X^i$ are hermitian. Then
the matrices $\{X^1,X^2,X^3\}$ can be interpreted geometrically as quantized embedding functions
\be
\begin{pmatrix}
 X^1 + i X^2 \\ X^3
\end{pmatrix}
 \sim \begin{pmatrix} R e^{i y_3} \\ x^3
\end{pmatrix}: \quad  S^1 \times \R \hookrightarrow \R^3 .
\ee
This defines  the fuzzy cylinder $S^1 \times_\xi \R$. It is 
the quantization of $T^* S^1$ with canonical Poisson bracket $\{e^{iy_3},x^3\} = -i\xi e^{i y_3}$, i.e.
 $\{x^3,y^3\} = \xi$ locally.

\paragraph{Wave-functions.}

A basis of functions on $S^1 \times_\xi \R$ is given by 
\be
\{e^{i p X^3} U^n, \quad p \in [-\frac{\pi}\xi,\frac{\pi}\xi], \quad n\in \Z \} ,
\label{cylinder-functions}
\ee
so that the most general function on $S^1 \times_\xi \R$  i.e. matrix $\phi \in End(\cH)$ can be expanded as
\be
\phi = \sum_{n\in\Z} \int_{-\pi/\xi}^{\pi/\xi} d p \,\tilde \phi_n(p)\, e^{i p X^3} U^n .
\label{cylinder-decomp}
\ee
Note that the set of linear momenta $p$ is in fact compactified on a circle.
This follows from
$U e^{i p X^3} = e^{i p \xi}\, e^{i p X^3} U $, so that
\be
e^{i p X^3} \equiv e^{i (p + \frac{2\pi}\xi) X^3}
\ee
as operators on $\cH$. 
This observation is very important: it means that there is an effective UV cutoff in the momentum space
for $\R$. This is a  consequence of the uncertainty relations combined with split noncommutativity: 
since the compact space has an IR cutoff $\frac 1R$, the non-compact space has a UV cutoff $\L_{UV} = \xi^{-1}$.
This is physically very welcome, and in sharp contrast to non-compact noncommutative spaces such as the Moyal-Weyl 
$\R^2_\theta$ which has {\em no} intrinsic UV cutoff in spite of the 
uncertainty relations.  On the other hand, there is no cutoff in the winding modes $n$.
A related observation has been made in \cite{Chaichian:1998kp}.
In particular, the space of all functions in \eq{cylinder-decomp} 
(and the spectrum of the Laplacian \eq{spec-laplace-cylinder}) has the characteristics of a one-dimensional space, 
as in a 1-dimensional QFT. This is also consistent with
the relation $2\pi R\xi \,{\rm Tr} \sim {\rm Vol} (\cM)$, which has a 1-dimensional volume divergence.
The relevance of these observations to the noncommutative gauge theory on $\cM$ 
will be discussed in section \ref{sec:gauge}.

\paragraph{Matrix Laplacian.}

The matrix equations of motion \eq{eom-IKKT}, \eq{eom-BFSS} are governed by the following matrix Laplace operator
\be
\Box := [X^a,[X^b,.]] \eta_{ab} \quad \sim \, e^{\sigma} \Box_G .
\ee
We note the following useful identity
\bea
2[X^1,[X^1,\phi]] + 2[X^2,[X^2,\phi]] &=& [Z,[Z^\dagger,\phi]] +  [Z^\dagger,[Z,\phi]] \nn\\
  &\stackrel{[Z,Z^\dagger] =0}{=}& 2 [Z^\dagger,[Z,\phi]]  \, = \,2 [Z,[Z^\dagger,\phi]] 
\label{id-2}
\eea
where $Z = X^1+iX^2$. 
For the fuzzy cylinder algebra \eq{fuzzy-cylinder}, this implies
\bea
\Box X^3 &=& [X^1,[X^1,X^3]] + [X^2,[X^2,X^3]] 
 = [U,[U^\dagger,X^3]] = 0  \nn\\ 
\Box X^1 &=& [X^2,[X^2,X^1]] + [X^3,[X^3,X^1]] \nn\\
 &=& [X^3,[X^3,X^1]] = -i\xi [X^3,X^2] = \xi^2 X^1, \nn\\
\Box X^2 &=& \xi^2 X^2
\label{Laplace-cylinder-explicit}
\eea
i.e. 
\be
\Box U = \xi^2 U, \quad \Box X^3 = 0 .
\label{Laplace-cylinder-U}
\ee
Thus $X^3$ is ``harmonic'' while $X^1, X^2$ are in some sense ``massive''. 
Therefore the fuzzy cylinder is not a solution of either the IKKT or the BFSS matrix model. 
However, it is 
quite obvious how to build a corresponding solution for the BFSS model:
the cylinder should be rotating.

\subsection{Rotating cylinder solutions}
\label{sec:rotating-cylinder}

\paragraph{BFSS solution.}
 
Starting with a fuzzy cylinder $(U,X) \sim (e^{iy}, x)$ as above with NC modulus $\xi$ and radius $R$, 
define the following 3 time-dependent matrices
\be
\begin{pmatrix}
 X^1(t) + i X^2(t) \\ X^3(t)
\end{pmatrix} \,\,
= \begin{pmatrix}
   Ue^{i\xi t}  \\ X^3
  \end{pmatrix} 
 \sim\,\,  \begin{pmatrix}
 R e^{i  (y+\xi t) } \\ x 
\end{pmatrix}  .
\label{BFSS-cylinder}
\ee
It is obvious using \eq{Laplace-cylinder-U} that this gives a solution of the BFSS matrix model, 
which is well-known \cite{Bak:2001kq}.
In the semi-classical limit, 
this matrix geometry describes
$\cM \sim \R_t \times S^1 \times  \R_x$. This is a D2-brane solution
which is stabilized because the $S^1$ is rotating, extended along an arbitrary $X^3$ direction.

\paragraph{IKKT solution.}

Next we want to find a corresponding solution of the more geometric IKKT model.
One cannot apply the same trick directly, since there is no commutative time.
However based on general arguments \cite{Ishibashi:1996xs}, there should be a corresponding solution.

Consider a fuzzy cylinder $(U,X) \sim (e^{iy}, x)$ as above with NC modulus $\xi$  and radius $R$.
We can embed the non-compact direction along a light-like direction $v$. For example, 
consider the matrices $X^a, a=0,1,2,3$ defined as
\be
\begin{pmatrix}
 X^0\\  X^1 + i X^2 \\ X^3
\end{pmatrix} \,\,
= \begin{pmatrix}
   0\\   U \\ 0
  \end{pmatrix} + v^a\, X
 \sim\,\,   \begin{pmatrix}
 x \\ R e^{i y } \\ x 
\end{pmatrix} , \qquad v^a = \begin{pmatrix}
 1 \\  0 + i 0 \\ 1
\end{pmatrix} .
\label{rotating-cylinder}
\ee
where $X^0$ is the time-like direction.
This is indeed a solution of the IKKT model 
\be
\Box X^a = 0 \qquad \mbox{for} \quad v^a v^b \eta_{ab} = 0 
\ee
for any $\xi$ and $R$.
Note that this works only in the Minkowski case.
The semi-classical limit is given by the geometry $S^1 \times \R$ with Poisson structure
$\{x, R e^{i y} \} = i\xi  R e^{i y}$.
This defines the {\em propagating fuzzy cylinder}, which
is propagating in a light-like direction.  
The induced metric in the $(x,y)$ coordinates is $g_{\mu\nu} = \diag(0,R)$,
and the effective metric \eq{eff-metric} is $G^{\mu\nu} \sim \diag(R,0)$.

Several remarks are in order. First, note that  the induced metric $g_{\mu\nu}$ is degenerate.
This means that $\int d^2 x\,\sqrt{g}$ vanishes identically, i.e. there is no 
``cosmological constant''. However this applies only to the 2-dimensional case,
and the higher-dimensional generalizations below will
have non-degenerate metrics.

Furthermore, note that we obtained a compactification without adding any
other terms (such as cubic terms) and scales to the matrix model. Such additional terms would necessarily break the
$SO(9,1)$ symmetry of the model, and strongly constrain the geometry of the non-compact 
space\footnote{unless 
the fluxes arise purely dynamically, perhaps through fermion condensation; however 
such a mechanism  has not been established.}. This would be in conflict with  gravity.
Accordingly, the radius of $S^1$ as well as $\xi$ are free moduli, 
and not determined by some explicit scale or potential in the action. 
This aspect will be discussed further in section \ref{sec:moduli}.

This cylindrical solution of the IKKT model can be interpreted as a closed (D-) string. 
However, it is quite different from the BFSS solution \eq{BFSS-cylinder}: 
 the cylinder is propagating along a light-like direction  
while the $S^1$ is essentially constant. We will obtain different types of solutions below.

\subsection{Propagating plane wave solution}

The following simple solution of the IKKT model describes a 2-dimensional plane wave
which propagates along a non-compact time direction.
We first make a trivial but useful observation: 
If $[\bar X^\mu,\bar X^\nu] = i \theta^{\mu\nu}$ generate the quantum plane,
then the two matrices $(e^{i k_\mu \bar X^\mu}, \bar X^\eta)$ (for fixed $\eta$) satisfy the relations 
of a fuzzy cylinder \eq{fuzzy-cylinder}, with $R=1$ and NC modulus $\xi = - k_\mu \theta^{\mu \eta}$.

Now let $[\bar X^3,\bar X^0] = i \theta$ generate the quantum plane $\R^2_\theta$, and 
define 4 hermitian matrices as follows:
\begin{align}
\begin{pmatrix}
 X^0 \\ Z \\ X^3 
\end{pmatrix}
:= \begin{pmatrix}
 \bar X^0 \\ R\, e^{i\frac \xi\theta(\bar X^0 + \bar X^3)} \\ \bar X^3 
\end{pmatrix} 
\label{prop-plane-wave}
\end{align}
where $X^0$ is the time-like direction and $Z = X^1 + i  X^2$.
This is reminiscent of \eq{rotating-cylinder} except that 
$X^0$ and $X^3$ no longer commute. It is again easy to see that 
\be
\Box Z = 0, \qquad \Box X^0 = \Box X^3 = 0,
\ee
thus we obtained a solution of the IKKT model.
This defines the {\em propagating plane wave}.  
The semi-classical limit is given by the geometry $\R^2 \hookrightarrow \R^4$
with a plane-wave-like embedding, and  Poisson structure
$\{x^0, x^3 \} = \theta$.
The induced metric $g_{\mu\nu}$ is given by
\bea
g_{\mu\nu} &=& \eta_{\mu\nu} + \frac 12 \del_\mu \bar z\, \del_\nu z +\frac 12 \del_\mu z\, \del_\nu  \bar z ,
\eea
which in light-cone coordinates $x^\pm = x^0 \pm x^3$ is 
\be
g_{\mu\nu} 
=   \begin{pmatrix}
 R^2 \xi^2/ \theta^2 & -\frac 12 \\ -\frac 12   & 0 \end{pmatrix} .
\ee
This is flat with Minkowski signature.
The effective metric $G^{\mu\nu}$ is then also flat and Minkowski, where
the role of time and space is switched. 

Although this solution has no compactified extra dimensions, we can use 
a similar construction to generate such solutions. 
 We will start with some solution of 
the space-like matrix equation $\Box' X^a = \xi^2 X^a$, and turn it into a rotating solution of the 
IKKT model. This will be discussed next, focusing on the case  of 4 non-compact directions.

\section{Higher-dimensional compactification}

\subsection{Stabilization by angular momentum}

Assume we have a matrix geometry in the $d$-dimensional Euclidean matrix model
which satisfies
\be
\Box_Y Y^i = \xi_d^2 Y^i, \qquad \Box_Y \equiv [Y^i,[Y^j,.]]\d_{ij}, \quad i,j = 1,..., d .
\ee
There are many explicit examples corresponding to quantized compact space $K$, such as
the fuzzy sphere $S^2_N \subset \R^3$ \eq{fuzzy-sphere}, the fuzzy torus $T^2_N \subset \R^4$ \eq{fuzzy-torus}, 
fuzzy $\C P^n_N \subset \R^{n^2+2n}$ \cite{Grosse:1999ci,Balachandran:2001dd,Grosse:2004wm}, and others.
We would like to obtain a 
corresponding solution of the IKKT or BFSS matrix model,
by giving  angular momentum to $K$.

In the time-dependent BFSS model, this can be achieved simply by 
assembling the hermitian matrices into complex ones 
\be
Z^\a = Y^{2\a-1} + i Y^{2\a},
\quad \bar Z^\a = Y^{2\a-1} - i Y^{2\a}, 
\ee
and giving them a time-dependence as follows
\be
Z^\a(t)':= Z^\a e^{i\omega t}, 
\qquad \omega^2 = \xi_d^2 .
\label{time-dep-matrices}
\ee
It is obvious that this solves the matrix equations of motion
$\ddot X^a + \Box_d X^a = 0$. 
Note that the rotation may or may not be a symmetry of $K$. 
Since $a=1,...,9$ in the BFSS model, only $K\subset \R^8$ can be rotated in this way.
If we want to have solutions with the topology $\R^4 \times K$ in the IKKT model, then only $K \subset \R^6 \cong \C^3$
is admissible.

Now we describe two simple constructions which provide 
similar solutions of the IKKT model.

\subsection{Twisting via a fuzzy cylinder}

\begin{lemma} 
Suppose $(X^\mu,Y^i)$, $i=1,...,6$  are hermitian matrices which satisfy
\begin{align}
\,[X^\mu,[X^\mu,Y^i]] &= (\xi^\mu)^2\, Y^i \qquad \mbox{(no sum over $\mu$)}  \nn\\
\, \Box_Y Y^j \equiv \sum_i [Y^i,[Y^i,Y^j]] &= \xi_Y^2 \, Y^j 
\label{config}
\end{align}
Let $(X,U)$ be a fuzzy cylinder \eq{fuzzy-cylinder} with radius $1$ and NC modulus $\xi_X$, 
which commutes with the above matrices. 
Collect the $Y^i$ into complex matrices as
\begin{align}
Z^\a &= Y^{2\a-1} + i Y^{2\a}, \nn\\
\bar Z^\a &= Y^{2\a-1} - i Y^{2\a} \qquad \,\a=1,2,3 
\label{complexify}
\end{align}
and assume $[Z^\a,(Z^\a)^\dagger] = 0$ for  $\a=1,2,3$.
Then the 6 hermitian matrices ${Y^i}'$ defined via
\begin{align}
\begin{pmatrix}
{Z^1}' \\
{Z^2}' \\
{Z^3}'
\end{pmatrix}
=
\begin{pmatrix}
 Z^1 \, U^{n_1} \\
 Z^2 \, U^{n_2} \\
 Z^3 \, U^{n_3} 
\end{pmatrix}
\label{U-action-def}
\end{align}
satisfy 
\begin{align}
\,[X^\mu,[X^\mu,{Y^i}']] &= (\xi^\mu)^2\, {Y^i}'  \qquad \mbox{(no sum over $\mu$)},  \label{lemma1-1}\\
\,[X,[X,{Z^\a}']] &= n_\a^2 \xi_X^2 \,{Z^\a}' \label{lemma1-2}\\
\Box_{Y'} {Y^j}' &= \xi_Y^2 \, {Y^j}'    \label{lemma1-3}\\
\, \Box_{Y'} X  &= 0   \label{lemma1-4}\\
\, \Box_{Y'} X^\mu  &= \, \Box_{Y} X^\mu 
\label{lemma1-5}
\end{align}
\end{lemma}
\begin{proof}
\eq{lemma1-1} and \eq{lemma1-2}  are immediate using the $Z^\a$ variables. \eq{lemma1-3} 
can be seen using the identity \eq{id-2}, e.g.
\bea
2[{Y^1}',{[Y^1}',{Z^\a}']] + 2[{Y^2}',[{Y^2}',{Z^\a}']] &=& [Z^1 U^{n_1},[{U^{-n_1}} {Z^1}^\dagger,{Z^\a}']] 
+  [{U^{-n_1}} {Z^1}^\dagger,[Z^1 U^{n_1},{Z^\a}']] \nn\\
&=& [Z^1 ,[{Z^1}^\dagger,{Z^\a}' ]]  +  [ {Z^1}^\dagger,[Z^1 ,{Z^\a}' ]]
\eea
because $U$ commutes with $Z^\a$ and ${Z^\a}'$. 
The heuristic reason is that the twisting amounts to an orthogonal 
transformation, which leaves the matrix Laplacian invariant.
The same computation with $X^\mu$ instead of ${Z^j}'$, e.g.
\be
[{Y^1}',[{Y^1}',X^\mu]] + [{Y^2}',[{Y^2}',X^\mu]] = [{Y^1},[{Y^1},X^\mu]] + [{Y^2},[{Y^2},X^\mu]]
\ee
gives \eq{lemma1-5}. 
Similarly, \eq{lemma1-4} follows from e.g.
\bea
2[{Y^1}',{[Y^1}',X]] + 2[{Y^2}',[{Y^2}',X]] &=& [Z^1 U^{n_1},[{U^{-n_1}} {Z^1}^\dagger,X]] 
+  [{U^{-n_1}} {Z^1}^\dagger,[Z^1 U^{n_1},X]] \nn\\
 &=& -n_1\xi_X[Z^1 U^{n_1},{U^{-n_1}} {Z^1}^\dagger] + n_1\xi_X [{U^{-n_1}} {Z^1}^\dagger,Z^1 U^{n_1}] \nn\\
&=& 0
\eea
using \eq{id-2} and \eq{fuzzy-cylinder}, because $Z^\a$  commutes with  the fuzzy cylinder $(U,X)$
and $[Z^\a,(Z^\a)^\dagger] = 0$.
\end{proof}

To understand the geometrical significance, assume that the original matrices $X^a = (X^\mu,Y^i)$
describe a quantized embedding $X^a \sim x^a:\,\,\cM \hookrightarrow \R^D$ of a symplectic
manifold $(\cM,\theta^{-1}_{\mu\nu})$.
Let $(U,X) \sim (e^{i y},x)$ be the fuzzy cylinder. 
 Then the above construction in the semi-classical limit amounts to a map
\begin{align}
 \R \times S^1  \times \cM &\rightarrow \cM' \subset \R^{D+1}  \nn\\
(x,e^{i y},p) &\mapsto (x,e^{i y}\cdot p)
\label{product-map}
\end{align}
where $e^{i y}\cdot p$ stands for the  $S^1$  action on $\cM$ corresponding to \eq{U-action-def},
and
\be
(X,X^\mu,{Y^i}') \sim (x,{x^a}'): \quad \cM' \hookrightarrow \R^{D+1}  
\label{new-matrix-embedding}
\ee
is a quantized embedding map for $\cM'$. 
The corresponding Poisson structure on $\cM'$ is the push-forward of the Poisson structures 
$\{e^{i y_3},x_3\} = -i\xi e^{i y_3}$ on $\R \times S^1$ and $\theta^{\mu\nu}$ on $\cM$
via \eq{product-map}. One must distinguish two cases. 
First,  if the action $e^{i y}\cdot p$ defines a flow on $\cM$, then 
$\cM'$ is odd-dimensional. The Poisson structure is thus degenerate, with
symplectic leaves labeled by the eigenvalues of some central function. 
We will give an example below.
Second, if the map \eq{product-map} is free i.e. (at least locally) a diffeomorphism
(hence $e^{i y}\cdot p$ does not preserve $\cM$), then the image becomes a symplectic manifold, 
with  quantized embedding map given by \eq{new-matrix-embedding}.

Note  that one can apply a $SO(D)$ transformation on the $Y^i$ before
defining the complex combinations \eq{complexify}; this will be exploited below.
Similarly,  the non-compact direction of the
added cylinder $\R\times S^1$ may be oriented along an arbitrary direction, which maybe
space-like, time-like, or light-like.
Finally, some of the matrices $Y^i$ are allowed to vanish $Y^i = 0$. 
This will be useful below.

\subsection{Twisting via a plane wave}

Recall that
if $[\tilde X^\mu,\tilde X^\nu] = i \theta^{\mu\nu}$ generate the quantum plane $\R^2_\theta$,
then the two matrices $(e^{ik_\mu\tilde  X^\mu}, \tilde X^\eta)$ (for fixed $\eta$) satisfy the 
relations
of a fuzzy cylinder \eq{fuzzy-cylinder}, with $R=1$ and NC modulus $\xi = - k_\mu \theta^{\mu \eta}$.
We thus obtain the following analog of lemma 1:

\begin{lemma} 
Suppose $(X^\mu,Y^i)$, $i=1,...,6$  are hermitian matrices which satisfy
\begin{align}
\,[X^\mu,[X^\mu,Y^i]] &= (\xi^\mu)^2\, Y^i \qquad \mbox{(no sum over $\mu$)}  \nn\\
\, \Box_Y Y^j \equiv \sum_i [Y^i,[Y^i,Y^j]] &= \xi_Y^2 \, Y^j 
\label{config-twist}
\end{align}
Let $[\tilde X^\mu,\tilde X^\nu] = i \tilde\theta^{\mu\nu},\,\, \mu,\nu = 0,1$ generate a quantum plane $\R^2_\theta$
which commutes with the above matrices, and let $U= e^{ik_\mu \tilde X^\mu}$.
Then the 6 hermitian matrices ${Y^i}'$ defined as in \eq{complexify}, \eq{U-action-def}
satisfy 
\begin{align}
\,[X^\mu,[X^\mu,{Y^i}']] &= (\xi^\mu)^2\, {Y^i}'  \qquad \qquad \mbox{(no sum over $\mu$)},  \label{lemma2-1}\\
\,\tilde\eta_{\mu\nu} [\tilde X^\mu,[\tilde X^\nu,{Z^\a}']] &= n_\a^2  \, (k \cdot k)\, {Z^\a}',
 \qquad \mu,\nu = 0,1 \label{lemma2-2}\\
\Box_{Y'} {Y^j}' &= \xi_Y^2 \, {Y^j}'    \label{lemma2-3}\\
\, \Box_{Y'} \tilde X^\mu  &= 0 , \qquad 
\, \Box_{Y'} X^\mu  = \, \Box_{Y} X^\mu 
\label{lemma2-5}
\end{align}
provided $[Z^\a,(Z^\a)^\dagger] = 0$ for  $\a=1,2,3$.
Here
\be
 k \cdot k := \tilde G^{\mu\nu} k_\mu k_\nu   ,
 \qquad \tilde G^{\mu\nu} = \tilde\theta^{\mu\m'}\tilde\theta^{\nu\nu'}\tilde\eta_{\mu'\nu'} .
\ee
\end{lemma}
\begin{proof}
This follows easily from lemma 1 applied to the fuzzy cylinder algebras $(e^{ik_\mu\tilde  X^\mu}, \tilde X^\eta)$. 
\end{proof}

The geometrical significance of this construction is clear:   $(\tilde X^\mu, X^\mu,Y^i)$
describes a quantized embedding 
$\R^2_\theta \times \cM \hookrightarrow \R^D$ with non-trivial Poisson structure along some 
non-compact direction $\R^2_\theta$, 
and the modified embedding ${Y^i}'$ induces a rotation  
along this $\R^2_\theta$. 

With these constructions at hand, we can obtain new solutions of the IKKT model
with compact extra dimensions stabilized by angular momentum, analogous to \eq{time-dep-matrices}.
While adding a cylinder will lead to un-desired closed time-like circles, adding the 
plane waves will give the desired compactifications which propagate along a non-compact 
direction.

\subsection{Higher-dimensional cylindrical solutions}
\label{sec:higher-cylinder}

\paragraph{BFSS solutions $\R_t \times \R^n \times_\xi T^n$ for $n \leq 3$.}

As an example, we can take 3 mutually commuting copies of the fuzzy cylinder 
$S^1 \times_{\xi_i} \R$ realized by $(X^i,U_i)$ for $i=1,2,3$ as 
defined above, and give the $U_i$ a time-dependent factor $e^{i\omega_i t}$ with $\omega_i^2 = \xi_i^2$.
Thus define 9 hermitian time-dependent matrices as follows
\begin{align}
\begin{pmatrix}
 X^i(t), \,\, i = 1,2,3 \\
Z^1(t) \\
Z^2(t) \\
Z^3(t)
\end{pmatrix}
=
\begin{pmatrix}
 X^i ,\,\, i = 1,2,3 \\
 U_1 e^{i \omega_1 t} \\
U_2 e^{i \omega_2 t}\\
U_3  e^{i \omega_3 t}
\end{pmatrix}
\end{align}
where $Z^1 = X^4+iX^5$ etc.
This clearly gives a solution of the BFSS matrix model
of type $\R_t \times \R^3 \times_\xi T^3$, interpreted as 3 rotating cylinders. 
Note that the noncommutative structure is indeed split 
as discussed previously, and the $\R^4$ subspace is commutative. 

There are obviously many variations of this solution,
such as $\R_t \times \R^2 \times_\xi T^2$, or with  different winding numbers by replacing $U_i \to U_i^{n_i}$
compensated by $\omega_i \to n_i \omega_i$.
It is also possible to replace one  $T^2$ by a noncommutative torus $T^2_N \subset \R^4$ \eq{fuzzy-torus}, and obtain e.g.
$\R_t \times \R \times_\xi S^1 \times T^2_N$.

\paragraph{Non-propagating IKKT solution $\R^n \times_\xi T^n$ for $n\leq 3$.}

Now we want to construct similar solutions of the IKKT model, based on lemma 1. 
Consider again 3 mutually commuting fuzzy cylinders $S^1 \times_\xi \R$ with NC modulus $\xi_i$
and radius $R_i$,
realized by $(X^i,U_i)$ for $i=1,2,3$,
embedded along space-like directions:
\begin{align}
\begin{pmatrix}
 X^0 \\
 X^i, \,\, i = 1,2,3 \\
Z^1 \\
Z^2 \\
Z^3
\end{pmatrix}
=
\begin{pmatrix}
 0 \\
 X^i ,\,\, i = 1,2,3 \\
 U_1  \\
U_2  \\
U_3 
\end{pmatrix}
\label{3-cylinders}
\end{align}
where $Z^1 = X^4+iX^5$ etc. Now we add  a ``time-like''  fuzzy cylinder $(U_0,X^0) \sim (e^{i y_0},x_0)$ 
with NC modulus $\xi_0$,
which commutes with the remaining generators. Thus define 10 new
hermitian  matrices as follows
\begin{align}
\begin{pmatrix}
 {X^0}' \\
 {X^i}', \,\, i = 1,2,3 \\
{Z^1}' \\
{Z^2}' \\
{Z^3}'
\end{pmatrix}
=
\begin{pmatrix}
 X^0 \\
 X^i ,\,\, i = 1,2,3 \\
 U_1 U_0 \\
U_2 U_0 \\
U_3 U_0
\end{pmatrix}
\end{align}
Now  lemma 1 can be applied along with \eq{Laplace-cylinder-explicit}, which implies that
\be
\Box' {Z^i}' = (\xi_i^2 - \xi_0^2) {Z^i}', \qquad \Box' {X^\mu}' = 0
\ee
Therefore we obtain a solution of the IKKT model for $\xi_0^2 = \xi_i^2$. 
Solutions with different winding numbers $n_i$ can be obtained by adjusting the $\xi_i$ accordingly.
However, according to the discussion below lemma 1 there is a constraint.
The central generator is easily identified as
\be
{X^0}'-{X^1}'-{X^2}'-{X^3}' = C .
\ee
Therefore the symplectic leaves define 
D5-branes\footnote{In accord with the string literature we denote $n+1$ - dimensional submanifolds 
with Minkowski signature as $Dn$ branes.} with the structure $\cM \sim \R^3 \times T^3$, compactified along
$C=const$. In particular, the non-compact space $\R^4$ is completely isotropic.
 However, such irreducible solutions can be obtained more directly:

\paragraph{Non-propagating IKKT solution $\R^4 \times_\xi S^3 \times S^1$.}

We modify the above construction in order to avoid the degenerate Poisson structure. 
Starting again with \eq{3-cylinders} and coinciding $\xi_i \equiv \xi$, we define
\be
 \tilde Z^1 = X^4 + i X^6, \quad \tilde Z^2 = X^5 + i X^7,  \quad \tilde Z^3 = X^8 + i X^9
\label{permutation}
\ee
(note that $X^5$ and $X^6$ have been interchanged).
This amounts to an orthogonal transformation among the $X^4,...,X^9$.
Then clearly the new $X^1,..., X^9$ still
define $\R^3 \times T^3$ and satisfy $\tilde\Box \tilde X^a = \xi^2 \tilde X^a, a=4,...,9$ and 
$\tilde\Box \tilde X^i =0, \,\, i = 1,2,3$.
Now let them rotate again by 
adding a time-like cylinder $(U_0,X^0) \sim (e^{i y_0},x_0)$ 
with NC modulus $\xi_0$
as follows
\begin{align}
\begin{pmatrix}
 {X^0}' \\
 {X^i}', \,\, i = 1,2,3 \\
{Z^1}' \\
{Z^2}' \\
{Z^3}'
\end{pmatrix}
=
\begin{pmatrix}
 X^0 \\
 X^i ,\,\, i = 1,2,3 \\
 \tilde Z^1 U_0 \\
\tilde Z^2 U_0 \\
\tilde Z_3 U_0
\end{pmatrix} .
\label{solution-T4}
\end{align}
As above, this is a solution $ \Box' {X^a}' = 0$ provided $\xi_0 = \xi$.
The point is that now the time-like  rotation does not preserve $T^3$, but 
sweeps out $S^3 \times S^1$. Therefore
the above matrices define a quantized $\R^4 \times S^3\times S^1$,
with classical coordinates $x^\mu,y_\nu$ and Poisson structure $\{x^\mu,e^{iy_\nu}\} = \xi\,\delta^{\mu}_{\nu} e^{iy_\nu}$.
In particular, the non-compact space $\R^4$ is completely isotropic.

Now consider the induced metric, which in the $\zeta^A = (x^\mu,y_\nu)$ coordinates is given by 
\be
g_{AB} = \begin{pmatrix}
              \eta_{\mu\nu} & 0 \\
              0 & R^2\d_{\mu\nu}
             \end{pmatrix} .
\ee
Therefore the effective metric \eq{eff-metric}  is
\be
G^{AB} \sim \theta^{A A'} \theta^{BB'} g_{A'B'} 
  = \xi^2 \begin{pmatrix}
              R^2\d_{\mu\nu} & 0 \\
              0 & \eta_{\mu\nu}
             \end{pmatrix} .
\label{eff-metric-T4}
\ee 
This has indeed Minkowski signature, however 
the time-like direction is now in the compact space. This change of causality structure is a typical phenomenon 
in the context of emergent gravity (which should also occur for the open string metric in similar contexts). 
One possibility to avoid this is 
to consider complexified Poisson structures corresponding 
to complexified matrices, as discussed in \cite{Steinacker:2010rh}. However in the present paper, we 
insist that all $X^a$ are hermitian matrices, so that \eq{eff-metric-T4} must be taken serious.
In that case, the time-like directions are compactified, and there is no propagation along the 
non-compact space $\R^4$. This will apply in particular for the lowest Kaluza--Klein modes. 
Therefore these solutions are interesting but unphysical, and we must look for solutions 
with Minkowski signature on the non-compact space. Such solutions will be found below,
using a twist along a quantum plane.

There are obviously many variations of this solution. By letting some cylinders degenerate
we obtain solutions $\R^n \times_\xi T^n$ for $n =2,3$. In particular, 
the reduced form of \eq{3-cylinders} can be recovered in this way. 
We can also introduce different winding numbers $U_i^{n_i}$ provided the
$\xi_i$ are adjusted accordingly. 
Finally, it is instructive to note that one can also use the semi-classical result 
$\Box \sim e^\sigma \Box_G$ \cite{Steinacker:2010rh} to see that \eq{solution-T4} is a 
(semi-classical) solution of the model.

\subsection{Propagating cylindrical IKKT solutions.}
\label{sec:propagating-cylinder}

In order to obtain an effective metric $G^{\mu\nu}$ which has Minkowski signature in 
the non-compact directions, we will use the construction in lemma 2 with Minkowski signature on
$\R^2_\theta$.

\paragraph{Propagating $\R^3 \times S^1$ and $\R^4 \times T^2$.}

As a first example, we start with a fuzzy cylinder $(U,X^2)$ with NC parameter $\xi$
and radius $R$, and twist it
with the noncommutative plane wave $[X^\mu,X^\nu] = i \theta^{\mu\nu}, \,\mu = 0,1$
(which commutes with the cylinder) as follows
\begin{align}
\begin{pmatrix}
 X^\mu, \,\, \mu = 0,1 \\
 X^2 \\
Z
\end{pmatrix}
=
\begin{pmatrix}
 X^\mu, \,\, \mu = 0,1 \\
 X^2 \\
U\, e^{i k_\mu X^\mu}
\end{pmatrix}
\label{R3T1-solution}
\end{align}
where $Z =  {X^4} + i {X^5}$.
Then lemma 2 gives
\begin{align}
\sum_{\mu,\nu = 0,1,2} \eta_{\mu\nu}[X^\mu,[X^\nu, Z]] &= 
 (\xi^2 +  k \cdot k)\, Z
\qquad\quad   k \cdot k := \tilde G_{(2)}^{\mu\nu} k_\mu k_\nu \nn\\
\sum_{i=4,5} \, [{X^i},[{X^i},X^\mu]]  &= 0, \qquad \mu = 0,1  
\end{align}
and $\sum_{i=4,5} \, [{X^i},[{X^i},X^2]]  = 0$ due to the fuzzy cylinder.
Here $\tilde G_{(2)}^{\mu\nu} = \theta^{\mu\mu'} \theta^{\nu\nu'} \eta_{\mu'\nu'}, \,\mu,\nu = 0,1$ 
denotes the 2-dimensional 
contribution to the effective metric.
Therefore we obtain a solution 
\be
\Box X^a = 0\quad \mbox{for} \qquad  k \cdot k = - \xi^2  .
\ee
In particular, $k$ must be time-like w.r.t. $\tilde  G_{\mu\nu}$.
These matrices define a quantization of $\R^3 \times S^1$,  parametrized by $\zeta^A = (x^\mu,x^2,y_2)$.
The Poisson structure and the induced metric are given by
\be
\{\xi^A,\xi^B\}  = \begin{pmatrix}
 \theta^{\mu\nu} & 0 & \tilde k^\mu  \\
 0 & 0 & \xi \\
 - \tilde k^\mu & - \xi & 0
\end{pmatrix}, \qquad 
g_{AB} = \begin{pmatrix}
              \eta_{\mu\nu} & 0 & 0 \\
              0 & 1 & 0 \\
              0 & 0 &  R^2
             \end{pmatrix} .
\ee
where
\be
\tilde k^\mu =  \theta^{\mu\nu} k_\nu .
\ee
Therefore the effective metric is
\be
 G^{AB} =  \theta^{A A'} \theta^{BB'} g_{A'B'} 
  =  \begin{pmatrix}
           \tilde G_{(2)}^{\mu\nu} + R^2\tilde k^\mu \tilde k^\nu & -R^2\xi \tilde k^\mu & \tilde G_{(2)}^{\mu\nu} k_\nu   \\
            -R^2\xi \tilde k^\nu & R^2\xi^2 & 0 \\
            \tilde G_{(2)}^{\mu\nu} k_\nu  &  0 & \tilde k^\mu\tilde k^\nu \eta_{\mu\nu} + \xi^2
             \end{pmatrix}  . \nn\\
\label{eff-metric-R3T1}
\ee 
Now the effective metric restricted to the non-compact $\R^3$  is given by
\be
G^{AB}_{\R^3} =  R^2 \begin{pmatrix}
           R^{-2}\tilde G_{(2)}^{\mu\nu} + \tilde k^\mu \tilde k^\nu & -\xi \tilde k^\mu \\
            -\xi \tilde k^\nu & \xi^2 
             \end{pmatrix} ,
\label{horiz-metric-3}
\ee
which has Minkowski signature, cf. \eq{eff-3D-prop-diag}. 
This is indeed the metric which governs the lowest KK modes, 
as shown in section \ref{sec:spec-Mink-cylinder}.
Therefore this solution can serve as 
physical space-time with compactified extra dimensions.

An analogous construction using 2 cylinders gives a solution of type 
$\R^2_\theta \times \R^2 \times_\xi T^2$ in terms of  8 matrices, with Minkowski signature in the 
non-compact space.

\paragraph{Propagating $\R^4 \times S^3 \times S^1$.}

In order to get solutions with 4 compact and 4 non-compact dimensions, 
we can add a compact fuzzy space such as $T^2_N$ or $S^2_N$ and then twist the whole 
construction.

We consider a variant of the construction in \eq{solution-T4}, starting with a  
fuzzy cylinder  $(U_2,X^2)$ with NC parameter $\xi$,
and a fuzzy torus $U,V$ as in \eq{fuzzy-torus} with $\xi^2 = 4 \sin^2(\frac\pi N)$. Thus define 
\begin{align}
X^a = \begin{pmatrix}
 X^2 \\
Z^1 \\ 
Z^2 \\
Z^3 \\
\end{pmatrix}
=
\begin{pmatrix}
 X^2 \\
U_2\, \\
U  \\
V  
\end{pmatrix}
\end{align}
which satisfies $\Box Z^\a = \xi^2 Z^\a$ using \eq{Box-torus}, and $\Box X^2 = 0$.
Now apply again an orthogonal transformation $Z^\a \to \tilde Z^\a$
as in \eq{permutation}, 
\be
 \tilde Z^1 = X^4 + i X^6, \quad \tilde Z^2 = X^5 + i X^7,  \quad \tilde Z^3 = X^8 + i X^9
\label{permutation-2}
\ee
(interchanging $X^5$ and $X^6$). We then add 
another fuzzy cylinder $(U_3,X^3)$ with NC modulus $\xi_3$, and finally twist the compact components
with a $\R^2_\theta$ with Minkowski signature:
\begin{align}
{X^a}' = \begin{pmatrix}
 {X^\mu}',\,\, \mu = 0,1  \\
 {X^i}' ,\,\, i = 2,3 \\
{Z^1}' \\
{Z^2}' \\
{Z^3}'
\end{pmatrix}
=
\begin{pmatrix}
 X^\mu \\
 X^i ,\,\, i = 2,3 \\
 \tilde Z^1 \, U_3 \, e^{i k_\mu X^\mu} \\
\tilde Z^2 \, U_3 \, e^{i k_\mu X^\mu} \\
\tilde Z_3\,  U_3\, e^{i k_\mu X^\mu} 
\end{pmatrix} .
\label{solution-T4-twisted}
\end{align}
Now there are no degeneracies due to the permutation \eq{permutation-2}, therefore
these matrices define a matrix quantization of $\R^4 \times S^3\times S^1$. 
Using lemmas 1 and 2, they provide a solution of the IKKT model 
for
\be
\Box' X^a = 0 \qquad \mbox{if} \quad k \cdot k = - (\xi^2 + \xi_3^2)  .
\ee
We can of course omit the last cylinder and obtain a solution $\R^3 \times T^3$.
Furthermore, it is clear that the fuzzy torus in the above construction can be replaced by a fuzzy sphere,
which gives a propagating IKKT solution of type  $\R^4 \times T^2 \times S^2_N$.
This will be given explicitly  in the next section.

\subsection{Spherical extra dimensions}

We finally provide some illustrative solutions with fuzzy spheres in the extra dimensions. 
Solutions with fuzzy spheres have been obtained up to now only upon
adding extra terms to the  matrix models, notably certain cubic terms; this is clearly 
undesirable. We show here how extra-dimensional fuzzy spheres can arise as solutions of the 
un-modified IKKT  model.

\paragraph{BFSS solution $\R_t \times \R^2_\theta \times \R\times_\xi S^1 \times S^2\times S^2$.}

Consider  2 mutually commuting fuzzy spheres  \eq{fuzzy-sphere}
 $S^2_L \times S^2_R \subset \R^6$ realized as follows 
\begin{align}
\, [Y^a_L, Y^b_L] &= i c_{L} \varepsilon^{abc} Y^c_L , \qquad \sum_a Y^a_L Y^a_L = R_L^2,  \nn\\
\, [Y^a_R, Y^b_R] &= i c_{R} \varepsilon^{abc} Y^c_R , \qquad \sum_a Y^a_R Y^a_R = R_R^2, \nn\\
\, [Y^a_L, Y^b_R] &= 0 ,
\label{2-spheres}
\end{align}
interpreted as quantized embedding maps $S^2_L \times S^2_R \hookrightarrow \R^3 \oplus \R^3$.
For a suitable choice of parameters these 6 matrices satisfy
\be
\Box_Y Y^a = \xi^2 Y^a .
\ee
This can be turned into a rotating solution of the BFSS model of type 
$\R_t \times S^2 \times S^2$ as above, 
by imposing a rotation $e^{i \omega t}$ in $\R^6$. 
However, we would  like to add  3 more non-compact directions $\R^3$.
One possibility is to add a (commuting) copy of the fuzzy cylinder
 $S^1 \times_\xi \R$ realized by $(X^3, U_3)$,  and an $\R^2_\theta$.  
This can be done as follows
\begin{align}
\begin{pmatrix}
 X^i(t), \,\, i = 1,2 \\
X^3(t) \\
Z^a(t)
\end{pmatrix}
=
\begin{pmatrix}
 X^i ,\,\, i = 1,2 \\
X^3 \\
 (Y^a_L + i Y^a_R) U_3 e^{i \omega t} \\
\end{pmatrix}
\end{align}
where $[X^1,X^2] = i \theta \one$ and 
with $\omega^2 = 2\xi^2$.
Since $[Y^a_L,Y^b_R]=0$, it follows as before that 
$(Z^a,X^3)$ is a fuzzy cylinder algebra for each $a=1,2,3$, and  
\begin{align}
\Box X^i = 0, \quad \mbox{for} \quad  i=1,2,3 , \qquad
\Box X^a = 2\xi^2 X^a\qquad \mbox{for} \quad  a=4,5,6,7,8,9 .
\end{align}
Intuitively, the matrices  $(Y^a_L + i Y^a_R) U_3$
still define $S^2_N \times S^2_N$ up to a  $U(1)$ rotation in $\C^3 \cong \R^6$.
Therefore the  above $X^a(t)$ are a solution of the BFSS equations of motion.

By setting the generators of one of the fuzzy spheres to zero, 
one  obtains a solution of the type 
$\R_t \times \R^2_\theta \times \R\times_\xi S^1 \times S^2$. 
Note  that it is not possible to use an internal symmetry of $S^2$ for the $U(1)$ associated with time,
because the corresponding vector fields have zeros and cannot stabilize $S^2$.

\paragraph{IKKT solution $\R^4\times T^2\times S^2$.}

To find an analogous solutions for the IKKT model, we can use the construction in \eq{solution-T4-twisted},
replacing a $T^2$ with an $S^2$. To do this, we first embed $S^2_N$ in $\R^6$ by adding a 
trivial extra coordinates to \eq{fuzzy-sphere}.
Thus consider
\begin{align}
X^a = \begin{pmatrix}
 X^\mu \\
Z^1 \\ 
Z^2 \\
Z^3 \\
\end{pmatrix}
=
\begin{pmatrix}
 X^\mu \\
 Y^1  \\
 Y^2  \\
 Y^3
\end{pmatrix}
\end{align}
where $Y^{1,2,3}$ form a fuzzy sphere with $\Box_Y Y^i= \xi^2 Y^i$. 
Now twist this configuration with two fuzzy cylinders $(U^2,X^2)$ and $(U_3,X^3)$ with NC moduli $\xi_2 = \xi_3$, 
and finally twist the compact directions with the quantum plane $\R^2_\theta$. This gives
\begin{align}
{X^a}' = \begin{pmatrix}
 {X^\mu}',\,\, \mu = 0,1  \\
 {X^i}' ,\,\, i = 2,3 \\
{Z^1}' \\
{Z^2}' \\
{Z^3}'
\end{pmatrix}
=
\begin{pmatrix}
 X^\mu \\
 X^i ,\,\, i = 2,3 \\
 Z^1 \, U_2 \, e^{i k_\mu X^\mu} \\
 Z^2 \, U_2 \, e^{i k_\mu X^\mu} \\
 Z_3\,  U_3\, e^{i k_\mu X^\mu} 
\end{pmatrix}
\label{solution-T2S2-twisted}
\end{align}
Using lemmas 1 and 2, this is a solution of the IKKT model provided
\be
\quad k \cdot k = - (\xi^2 + \xi_3^2) 
\ee
Note that  we need 2 noncommutative non-compact directions here, since no matrix
solution of type $\R^2 \times S^2$ with split noncommutativity 
(corresponding to $T^* S^2$) is known.

\paragraph{IKKT solution $\R^4\times S^2\times S^2$.}

Now start with the two commuting fuzzy spheres \eq{2-spheres}, 
and define $\tilde Z^\a = Y^\a_L + i Y^\a_R$.
Then
\begin{align}
{X^a} = \begin{pmatrix}
 {X^\mu},\,\, \mu = 0,...,3  \\
{Z^1} \\
{Z^2} \\
{Z^3}
\end{pmatrix}
=
\begin{pmatrix}
\bar  X^\mu \\
 \tilde Z^1 \, e^{i k_\mu\bar X^\mu} \\
 \tilde Z^2 \, e^{i k_\mu\bar X^\mu} \\
 \tilde Z^3 \, e^{i k_\mu \bar X^\mu} 
\end{pmatrix}
\label{solution-S2S2-twisted}
\end{align}
where $\bar X^\mu, \,\mu = 0,...,3$ define a 4-dimensional quantum plane $\R^4_\theta$.
Using an obvious generalization of lemma 2, it follows that
\be
\Box X^a = 0\qquad \mbox{for}\quad  k \cdot k = - \xi^2  
\ee
where $k \cdot k $ is now defined in terms of the 4-dimensional effective metric 
$\bar G^{\mu\nu}$ on $\R^4_\theta$.

\subsection{Moduli and deformations}
\label{sec:moduli}

It was pointed out above that the scales of the fuzzy cylinders are free moduli, 
and not determined by some explicit scale or potential in the action. 
Moreover, they also drop out from the action. To see this, 
assume that $(U,X^3)$ form a fuzzy cylinder, with $U = X^1 + i X^2$ and $[X^1,X^2] = 0$. Then
the potential (i.e. the argument of the Euclidean sector of the matrix model) can be written as
\bea
\sum_{i,j=1,2,3} [X^i,X^j][X^i,X^j] = [U,X^3][U^\dagger,X^3] + [U^\dagger,X^3][U,X^3]  = -2 \xi^2 R^2 ,
\label{action-cylinder}
\eea
while the opposite sign arises from the ``kinetic term'', where $X^3$ is replaced by the time-like $X^0$. 
Thus the the rotating cylinder solution \eq{rotating-cylinder} satisfies
\be
\sum_{\mu,\nu =0,..,3} [X^\mu,X^\nu][X^{\mu'},X^{\nu'}]\eta_{\mu\mu'}\eta_{\nu\nu'} = 0 .
\ee
In particular, the moduli $R$ and $\xi$ drop out from the action.
This can be understood
by recalling that the action is given semi-classically by $G^{\mu\nu}g_{\mu\nu}$
(up to normalization) \cite{Steinacker:2010rh}, together with the fact that $g_{\mu\nu}$ is degenerate as explained 
in section \ref{sec:rotating-cylinder}.
In the case of the propagating plane wave solution \eq{prop-plane-wave},
the same  computation applies, except that there is a non-vanishing contribution from the quantum plane 
$[X^0,X^1] = i \theta$.

%
%
%
%

Similar cancellations occur for the other solutions presented in this paper, 
where each fuzzy cylinder contributes a term \eq{action-cylinder}, which 
finally cancel with the time-like contribution. 
However, there are non-vanishing contributions from the explicitly noncommutative 
fuzzy tori or fuzzy spheres, and similarly from the NC planes.
It is nevertheless interesting to note that there is a cancellation mechanism, and
the moduli of the cylinders drop out from the action\footnote{the
rotating cylinder or ``supertube'' solution of the BFSS model is indeed known to be 
BPS \cite{Bak:2001kq}.} (but not from the energy). 
This suggests that such split NC solutions are flexible, in the sense that e.g. embedding deformations
may be accomodated by adjusting the  moduli. This should be important for gravity.


\subsubsection{2-dimensional solutions with arbitrary cross-section}
\label{sec:general-cross}

The propagating plane wave and rotating cylinder solutions \eq{prop-plane-wave}, \eq{rotating-cylinder}
can be generalized to generic shapes. This is well-known
for the BFSS case \cite{Terashima:2007uf}, and should not be too surprising 
in view of the discussion in the previous paragraph.
To see this, define the following ``light-cone matrix coordinates'' 
\be
X^\pm = X^0 \pm X^3 
\ee
which satisfy
\be
[X^+,X^-]  = 2 [X^3,X^0]  = -2i \theta .
\ee
Now for any function $f$ of a complex variable, define
\be
X^1+iX^2 := f(X^+) , \qquad X^1-iX^2 := \bar f(X^+) .
\ee
Then $[X^1,X^2] = 0$, and \eq{id-2} implies that 
\be
[X^1,[X^1,X^+]] + [X^2,[X^2,X^+]] = 0 = [X^1,[X^1,X^-]] + [X^2,[X^2,X^-]]
\ee
and similarly  
\be
[X^0,[X^0,U]] - 2[X^3,[X^3,U]] =0 .
\ee
Therefore
\be
\Box X^a = 0 .
\ee 
In the special case $X^{1} + i  X^{2} = R e^{-i\frac \xi h(X^0+X^3)}$, the propagating plane wave 
\eq{prop-plane-wave}  is recovered.
Alternatively, if we consider the sub-algebra generated by $X^-$ and $U= f(X^+)$, then  the 
deformed rotating cylinder solution \eq{rotating-cylinder} is recovered.
Thus we found  solutions which describes a deformed propagating wave,
and a rotating cylindroid.
This is analogous to the well-known chiral decomposition for a propagating string.
However, the corresponding generalization of the higher-dimensional solutions is not clear.

\section{Effective gauge theory and Kaluza-Klein modes}
\label{sec:gauge}

Now consider a background  of type $\R^n \times K$ in the matrix model where $K$ is compact,
with split noncommutativity as in \eq{split-NC}. It follows on general grounds 
that the fluctuations around such a 
background (more precisely: around a stack of coinciding such branes)
will define a noncommutative gauge theory on $\R^n\times K$.
Now by an important observation in section \ref{sec:split-NC}, this
implies --due to the uncertainty relations, assuming ``strictly'' split NC -- 
that there is a UV cutoff on $\R^4$ given by
\be
\L_{UV} = \L_{NC}^2 R
\label{UV-split}
\ee
where $R$ is the length scale of $K$.
This is physically very welcome, and in 
sharp contrast to non-compact noncommutative spaces such as the Moyal-Weyl quantum plane which
have no intrinsic UV cutoff.

It might appear that this would provide an easy cure for all UV divergences in QFT, 
but of course things are not that simple. 
To understand what happens, we first compute the spectrum of the Laplacian on the 
fuzzy cylinder:

\paragraph{Spectrum on fuzzy cylinder.}

Let us  evaluate the Laplace operator on the basis of functions \eq{cylinder-functions} on a 
fuzzy cylinder:
using
\begin{align}
[U,[U^\dagger,e^{i p X^3} ]] &= 2 R^2 e^{i p X^3} - U e^{i p X^3} U^\dagger - U^\dagger e^{i p X^3} U \nn\\
 &= R^2(2 - e^{-i p\xi} - e^{i p\xi})\, e^{i p X^3}= 4 R^2 \sin^2(p\xi/2)\, e^{i p X^3}
\end{align}
we obtain
\begin{align}
\Box e^{i p X^3} U^n &= [U,[U^\dagger,e^{i p X^3} U^n]] + [X^3,[X^3,e^{i p X^3} U^n]] \nn\\
 &= \(4 R^2 \sin^2(p\xi/2) + n^2 \xi^2 \)\,  e^{i p X^3} U^n \nn\\
 \, &\stackrel{p \xi \ll 1}{\sim} \, \(R^2 p^2 + n^2  \)\, \xi^2 e^{i p X^3} U^n 
\label{spectrum-cyl}
\end{align}
where $p \in [-\frac{\pi}{\xi},\frac{\pi}{\xi}]$.
Clearly $n$ labels the Kaluza-Klein (KK) modes on $S^1$. Indeed,
there is a UV cutoff on the non-compact space $|p|\leq \frac{\pi}{\xi} = \pi \L_{NC}^2 R$, as expected from 
\eq{UV-split}. This means that $\R$ is literally 
discretized by a lattice, as is manifest in the representation \eq{cylinder-rep}.
Analogous statements hold for higher-dimensional fuzzy cylinders $\R^n \times_\xi T^n$.

Therefore from a gauge theory point of view, 
the matrix model on the background $\R^n \times_\xi T^n$ behaves at low energies 
like an effective $n$-dimensional gauge theory on $\R^n$, with  effective UV cutoff
$\L_{UV} = \L_{NC}^2 R$. 
Thus the naive expectation of an  UV cutoff on NC spaces 
-- which is {\em not} borne out on $\R^n_\theta$ -- seems indeed realized here, moreover
in an essentially isotropic way without explicitly Lorentz-violating  $\theta^{\mu\nu}$ on $\R^n$.
Thus this split NC scenario appears to be very attractive.

However, there are problems. First, we have seen in section \ref{sec:higher-cylinder}
 that on such a ``strict'' split NC background, the 
non-compact space $\R^n$ has Euclidean signature 
(unless we resort to complexified backgrounds).  Furthermore, 
in the extreme UV limit the KK modes become relevant, and the
effective geometry becomes that of $T^n$,
as can be seen from \eq{spectrum-cyl}.
This leads to the same UV divergences as in a $n$-dimensional field theory. Moreover
UV/IR mixing may still occur, since for very low non-compact momentum $p$,
the internal loop momentum on $T^n$ may be arbitrarily high. 
Therefore in order to have a well-defined quantum theory, we still need maximal SUSY,
as realized in the IKKT model. We note the observation that NC backgrounds of type  $\R^4 \times K^4_N$ 
are effectively 4-dimensional in the UV.

It is interesting to compare this  with  product spaces such as $\R^4_\theta \times K^4_N$,
where $K^4_N$ is a $4$-dimensional fuzzy space described by a finite matrix algebra $\Mat(N,\C)$.
Then there is no UV cutoff on $\R^4_\theta$, while the tower of Kaluza-Klein modes on $K^4_N$ is 
finite. This is reversed compared with the above case of 
split noncommutativity, where $\R^4$ has a UV cutoff while $K$ has none.
However in either case, the UV divergences will be that of a 4-dimensional QFT,
either due to $\R^4$ or due to $K^4$.
Now recall that $\cN=4$ SYM is UV finite in 4 dimensions, for any genus in a 
large $N$ expansion. This strongly suggests that finiteness also holds for the noncommutative $\cN=4$ model,
i.e. for the IKKT model on such a background. 
This suggests that even D7 branes can be  consistent 
backgrounds for this model, provided 4 of their dimensions are compactified.
This should help to obtain rich enough structures for the compactification, towards a
realistic low-energy spectrum of the theory, 
cf. \cite{Aoki:2010gv,Aschieri:2006uw,Chatzistavrakidis:2009ix,intersect-paper}.

\subsection{Minkowski signature.}
\label{sec:spec-Mink-cylinder}

In the presence of Minkowski signature, the situation is once again more tricky.
Consider first the case of non-propagating cylinder solutions in section \eq{sec:higher-cylinder}.
For a fuzzy cylinder with time-like $X^0$, the Laplacian becomes
\begin{align}
\Box e^{i p X^0} U^n &= [U,[U^\dagger,e^{i p X^0} U^n]] - [X^0,[X^0,e^{i p X^0} U^n]] \nn\\
 &= [U,[U^\dagger,e^{i p X^0} ]] U^n - e^{i p X^0} [X^0,[X^0, U^n]] \nn\\
 &= \(4 R^2 \sin^2(p\xi/2) - n^2 \xi^2 \)\,  e^{i p X^0} U^n .
\label{spec-laplace-cylinder}
\end{align}
Note that the modes $U$ on the circle have negative sign. This is 
consistent with presence of closed time-like circles with respect to $G^{\mu\nu}$ as discussed before, 
which is clearly unphysical.

One possible way out might be to resort to some analytic continuation with purely imaginary $\xi$. 
Here we consider instead the propagating cylinder solutions in section \ref{sec:propagating-cylinder},
where these problems do not arise:

\paragraph{Propagating cylinders.}

Consider the propagating $\R^3_\theta \times T^1$ solution  \eq{R3T1-solution}. 
Then the wave- functions can be expanded in the following basis
\be
\{ e^{i p_j  X^j} U^n, \quad p_2 \in [-\frac{\pi}\xi,\frac{\pi}\xi], \,\, p_{0,1} \in \R, \quad n\in \Z \} .
\label{cylinder-functions-R3T1}
\ee
for $j=0,1,2$.
Now the Laplacian becomes
\begin{align}
\Box (e^{i p_j X^j} U^n) &= [\tilde U,[\tilde U^\dagger,e^{i p_j X^j} U^n]] 
  - [X^0,[X^0,e^{i p_j X^j} U^n]] 
+ \sum_{i=1,2} [X^i,[X^i,e^{i p_j X^j} U^n]] \nn\\
 &= \(4 R^2 \sin^2\Big(\frac{p_2\xi - k_\mu\theta^{\mu\nu} p_\nu}2\Big) 
+  \bar G_{(2)}^{\mu\nu} p_\mu p_\nu + n^2 \xi^2 \)\,  e^{i p_j X^j} U^n  
\label{spec-laplace-cylinder-M}
\end{align}
where  $\tilde U = U e^{i k_\mu X^\mu}$, and
$\bar G_{(2)}^{\mu\nu} = \theta^{\mu\mu'} \theta^{\nu\nu'} \eta_{\mu'\nu'}, \,\mu,\nu = 0,1$ denotes the 
$\R^2_\theta$ contribution to the effective metric.
Now the discrete Kaluza-Klein modes have positive mass as they should. Restricted to the lowest
KK sector $n=0$ and in the limit $p \xi \ll 1$, the spectrum of the Laplacian on 
the non-compact $\R^3$ becomes 
\begin{align}
\Box e^{i p_j X^j}  \,\approx 
   \,  \Big(R^2 \Big(\frac{p_2\xi - k_\mu\theta^{\mu\nu} p_\nu}2\Big)^2 
 + \bar G_{(2)}^{\mu\nu} p_\nu p_\nu\Big)\, e^{i p_\mu X^\mu}  
 = (p \cdot p)\, e^{i p_j X^j} , \qquad p \xi \ll 1
\label{eff-3D-prop-diag}
\end{align}
 where
\be
p \cdot p = R^2\,\begin{pmatrix}
             p_\mu, p_2
            \end{pmatrix}
\begin{pmatrix}
           R^{-2} \bar G_{(2)}^{\mu\nu} + \tilde k^\mu \tilde k^\nu & -\xi \tilde k^\mu \\
           - \xi \tilde k^\nu & \xi^2 
             \end{pmatrix} 
\begin{pmatrix}
             p_\mu \\ p_2
            \end{pmatrix} .
\ee
This  agrees with the semi-classical result using the effective metric 
\eq{horiz-metric-3}.
This has the desired physical properties, in particular  Minkowski signature. 
In the extreme UV, the spectrum is parametrized by $p_0, p_1$ and $n$,
corresponding to that of a 3-dimensional field theory on $\R^2\times S^1$.

Finally, one may consider analogous solutions of the undeformed $\cN=4$ SYM theory.
For example, fuzzy sphere configurations can be realized by the six $SU(N)$-valued scalar fields $\phi^i$, 
which correspond to the $Y^i$. 
Then configurations such as \eq{solution-S2S2-twisted} correspond to positive-energy solutions
of $\cN=4$ SYM, with a non-vanishing expectation value of some $U(1) \subset SU(4)$  generator of the 
internal $R$ - symmetry. Such a vacuum should be stable even though it has positive energy, 
because the charge is conserved. 
Note also that e.g. the flux of the $S^2_N$ is quantized and thus protected.


\subsection{Remarks on emergent gravity}
\label{sec:gravity}

In general, the geometry of brane solutions in the IKKT model is governed
by the $U(1)$ sector of the corresponding noncommutative gauge theory 
\cite{Steinacker:2007dq,Steinacker:2008ri,Steinacker:2010rh}.
The transversal fluctuations correspond to scalar fields, which via the bare matrix model action 
satisfy $\Box_G \phi^i = 0$. This does not directly lead (without quantum effects)
to general relativity, but might 
be relevant e.g. for large-scale cosmological modifications. On the other hand,
the tangential fluctuations  of the branes corresponding to trace--$U(1)$ gauge fields are governed 
by NC Maxwell equations, which -- as discovered by Rivelles \cite{Rivelles:2002ez} --
lead to Ricci-flat deformations of flat spaces $R_{\mu\nu}[\bar G + h] \approx R_{\mu\nu}[\bar G] = 0$.
Thus higher-dimensional branes should help to recover (near-) Ricci flat vacuum geometries and
general relativity from the matrix model, perhaps even without resorting to induced gravity.
Indeed $\R^4_\theta \times T^n$ is flat, hence the tangential $U(1)$ gauge fields 
should lead to $2+n$ degrees of freedom for Ricci-flat  on-shell deformations of the geometry.
This will be studied elsewhere in more detail.

It is also important in this context that 
the tubular solutions obtained here are flexible, because the scales are not fixed but free moduli.
Moreover, their action vanishes under suitable conditions, as discussed in section \ref{sec:moduli}.
All these aspects should be relevant for gravity, and may help to recover effectively GR in 4 dimensions.

\section{Conclusion}
\label{sec:conclusion}

In this paper, we constructed new solutions of the IKKT model of the type $\R^4 \times K$, where
$K= T^2,\, K= S^3\times S^1,\,  K = S^2 \times T^2$ and $K = S^2 \times S^2$. 
The compact spaces in these solutions are rotating, and  stabilized by angular momentum. 
This is in contrast to previous realizations of fuzzy spaces in matrix models, which required 
additional terms in the model that break some of the symmetries.
In particular, the Minkowski signature of the model is essential here, and there are no
such solutions in Euclidean models. This should shed new light on the search for 
non-perturbative vacua in the IKKT matrix model \cite{Aoyama:2010ry}.
Furthermore, these solutions are expected to be generic, in the sense that they admit
deformations both for the non-compact $\R^4$ part and also for the compact sector.
This is important from the point of view of the effective (emergent) gravity on the branes, which  
play the role of physical space-time.

These new solutions have several interesting features. 
They arise from a Poisson structure which connects the non-compact
with the compact space, dubbed ``split noncommutativity''. In contrast to compactifications on 
fuzzy spaces, this can lead to an infinite tower of
Kaluza-Klein modes, and a strict UV cutoff on (some directions of) the non-compact space.
In the IR limit the physics is governed by the 3+1--dimensional non-compact space. 
In the extreme UV, the physics is still 4-dimensional 
but with some compactified directions. In particular, the quantization is  expected to be finite even
on $\R^4 \times K^4$,
governed by NC $\cN=4$ SUSY Yang-Mills in 4 dimensions. 

Space-time solutions with compactified extra dimensions are clearly of great interest in the
context of particle physics, as a possible source of structure needed for realistic low-energy gauge theories.
A purely 4-dimensional $\R^4_\theta$ solutions of the IKKT model is too simple. 
However, extra dimensions may lead to SUSY breaking, and chiral fermions may arise e.g. on intersecting brane solutions. 
It is plausible that such intersecting branes 
should  have only 4 non-compact dimensions, and hence be compactified as above. 
Therefore the solutions presented here can be building blocks for realistic vacua in matrix models,
as demonstrated in \cite{Blumenhagen:2006ci,intersect-paper}.
%

Finally, since the geometry of the brane is governed by the $U(1)$ sector of the noncommutative gauge theory
 \cite{Steinacker:2007dq,Steinacker:2010rh}, extra dimensions also provide  
additional geometrical degrees of freedom and 
possibly new mechanisms, which may help to approximately recover general relativity
on space-time branes.



\section*{Appendix: De-compactification of fuzzy torus and sphere}

The fuzzy cylinder can be viewed as a de-compactification limit of other compact fuzzy spaces. 
Consider first the fuzzy torus $T^2_N$ embedded in $\R^4$. It is defined in terms of 
4 hermitian matrices packaged in terms of  
$U = X^1 + i X^2, \,\, V = X^3 + i X^4$ which satisfy the relations
\bea
(X^1)^2 + (X^2)^2 &=& 1 = (X^3)^2 + (X^4)^2 ,  \nn\\
\,[U, V] &=& (q-1) V U, 
\label{fuzzy-torus}
\eea
for $q= e^{2\pi i/N}$ and with $U^N = V^N = 1$, cf. \cite{de Wit:1988ig,hoppe}. 
The irreducible representations in terms of clock-and shift matrices
are well-known and need not be repeated here.
It is easy to see that the following relations hold
\be
\Box X^a 
 = 4\sin^2(\pi/N)\, X^a, \qquad a = 1,..., 4 .
\label{Box-torus}
\ee
If we ``de-compactify'' one circle by writing $V = \exp(i \a_N Y)$, then \eq{fuzzy-torus}
reduces to the fuzzy cylinder in the limit $N \to \infty$  with $\xi = \frac{2\pi}{N\a_N} = const$.

A similar de-compactification can also be carried out for  the fuzzy sphere
$S^2_N$ \cite{Madore:1991bw}. It is defined in terms of three $N \times N$ hermitian matrices $X^a, a=1,2,3$ 
subject to the relations
\begin{equation}
[ X^{{a}}, X^{{b}} ] = \frac{i}{\sqrt{C_N}}\varepsilon^{abc}\, X^{{c}}~ , 
\qquad \sum_{{a}=1}^{3} X^{{a}} X^{{a}} =  \one 
\label{fuzzy-sphere}
\end{equation}
where $C_N= \frac 14(N^2-1)$. 
It is easy to see that 
\be
\Box X^a = C_N\, X^a, \qquad a = 1,..., 4 .
\label{Box-sphere}
\ee
The fuzzy cylinder is obtained near the equator $X^3 \approx 0$, setting $U = X^1 + i X^2$
and upon appropriate rescaling for $N \to \infty$.


%
%

\paragraph{Acknowledgments.}

Useful discussions with O. Ganor, H. Kawai, D. L\"ust, S. Iso, Y. Kitazawa, 
J. Nishimura, N. Sasakura and P. Schreivogl are greatfully acknowledged, as well as hospitality at UC Berkeley, 
KEK and the Yukawa Institute.
I also thank A. Chatzistavrakidis for reading the manuscript.
This work  was supported by the Austrian Science Fund (FWF) under the contract
P21610-N16.

%



\begin{thebibliography}{99}

\bibliography{./mainbib}
\bibliographystyle{../custom2}


\bibitem{Ishibashi:1996xs}
  N.~Ishibashi, H.~Kawai, Y.~Kitazawa, A.~Tsuchiya,
  ``A Large N reduced model as superstring,''
  Nucl.\ Phys.\  {\bf B498 } (1997)  467-491.
  [hep-th/9612115].

\bibitem{Steinacker:2008ri}
  H.~Steinacker,
  ``Emergent Gravity and Noncommutative Branes from Yang-Mills Matrix Models,''
  Nucl.\ Phys.\  {\bf B810 } (2009)  1-39.
  [arXiv:0806.2032 [hep-th]]

\bibitem{Steinacker:2010rh}
  H.~Steinacker,
  ``Emergent Geometry and Gravity from Matrix Models: an Introduction,''
  Class.\ Quant.\ Grav.\  {\bf 27 } (2010)  133001.
  [arXiv:1003.4134 [hep-th]].

\bibitem{Seiberg:1999vs}
  N.~Seiberg, E.~Witten,
  ``String theory and noncommutative geometry,''
  JHEP {\bf 9909 } (1999)  032.
  [hep-th/9908142].


\bibitem{Blumenhagen:2006ci}
  R.~Blumenhagen, B.~Kors, D.~Lust, S.~Stieberger,
  ``Four-dimensional String Compactifications with D-Branes, Orientifolds and Fluxes,''
  Phys.\ Rept.\  {\bf 445 } (2007)  1-193.
  [hep-th/0610327].

\bibitem{Iso:2001mg}
  S.~Iso, Y.~Kimura, K.~Tanaka, K.~Wakatsuki,
  ``Noncommutative gauge theory on fuzzy sphere from matrix model,''
  Nucl.\ Phys.\  {\bf B604 } (2001)  121-147.
  [hep-th/0101102].
\bibitem{Myers:1999ps}
  R.~C.~Myers,
  ``Dielectric branes,''
  JHEP {\bf 9912 } (1999)  022.
  [hep-th/9910053].
\bibitem{Alekseev:2000fd}
  A.~Y.~.Alekseev, A.~Recknagel, V.~Schomerus,
  ``Brane dynamics in background fluxes and noncommutative geometry,''
  JHEP {\bf 0005 } (2000)  010.
  [hep-th/0003187].
\bibitem{Aschieri:2006uw}
  P.~Aschieri, T.~Grammatikopoulos, H.~Steinacker, G.~Zoupanos,
  ``Dynamical generation of fuzzy extra dimensions, dimensional reduction and symmetry breaking,''
  JHEP {\bf 0609 } (2006)  026.
  [hep-th/0606021].

\bibitem{Kimura:2001uk}
  Y.~Kimura,
  ``Noncommutative gauge theories on fuzzy sphere and fuzzy torus from  matrix
  model,''
  Prog.\ Theor.\ Phys.\  {\bf 106} (2001) 445
  [arXiv:hep-th/0103192].


\bibitem{Bak:2001kq}
  D.~Bak and K.~M.~Lee,
  ``Noncommutative supersymmetric tubes,''
  Phys.\ Lett.\  B {\bf 509} (2001) 168
  [arXiv:hep-th/0103148].


\bibitem{Janssen:2004jz}
  B.~Janssen, Y.~Lozano and D.~Rodriguez-Gomez,
  ``Giant gravitons in AdS(3) x S**3 x T**4 as fuzzy cylinders,''
  Nucl.\ Phys.\  B {\bf 711} (2005) 392
  [arXiv:hep-th/0406148].

\bibitem{Terashima:2007uf}
  S.~Terashima,
  ``Supertubes in matrix model and DBI action,''
  JHEP {\bf 0703 } (2007)  075.
  [hep-th/0701179].

\bibitem{Shepard:2005wy}
  P.~G.~Shepard,
  ``Bloch waves and fuzzy cylinders: 1/4-BPS solutions of matrix theory,''
  JHEP {\bf 0602 } (2006)  003.
  [hep-th/0510127].

\bibitem{Bak:2002wy}
  D.~s.~Bak, N.~Ohta, M.~M.~Sheikh-Jabbari,
  ``Supersymmetric brane - anti-brane systems: Matrix model description, stability and decoupling limits,''
  JHEP {\bf 0209 } (2002)  048.
  [hep-th/0205265].


\bibitem{Banks:1996vh}
  T.~Banks, W.~Fischler, S.~H.~Shenker, L.~Susskind,
  ``M theory as a matrix model: A Conjecture,''
  Phys.\ Rev.\  {\bf D55 } (1997)  5112-5128.
  [hep-th/9610043].


\bibitem{de Wit:1988ig}
  B.~de Wit, J.~Hoppe, H.~Nicolai,
  ``On the Quantum Mechanics of Supermembranes,''
  Nucl.\ Phys.\  {\bf B305 } (1988)  545

\bibitem{hoppe}
J.~Hoppe,``Membranes and matrix models,''
  [hep-th/0206192];
 J. Hoppe, "Quantum theory of a massless relativistic surface and a two-dimensional bound state problem``, 
 PH D thesis, MIT 1982


\bibitem{intersect-paper}
  A.~Chatzistavrakidis, H.~Steinacker, G.~Zoupanos,
  ``Intersecting branes and a standard model realization in matrix models,''
  [arXiv:1107.0265 [hep-th]].


\bibitem{Rivelles:2002ez}
  V.~O.~Rivelles,
  ``Noncommutative field theories and gravity,''
  Phys.\ Lett.\  {\bf B558 } (2003)  191-196.
  [hep-th/0212262]; 


\bibitem{Steinacker:2007dq}
  H.~Steinacker,
  ``Emergent Gravity from Noncommutative Gauge Theory,''
  JHEP {\bf 0712 } (2007)  049.
  [arXiv:0708.2426 [hep-th]].

\bibitem{Yang:2006mn}
  H.~S.~Yang,
  ``Instantons and Emergent Geometry,''
  Europhys.\ Lett.\  {\bf 88} (2009) 31002
  [arXiv:hep-th/0608013].


\bibitem{Steinacker:2008ya}
  H.~Steinacker,
  ``Covariant Field Equations, Gauge Fields and Conservation Laws from Yang-Mills Matrix Models,''
  JHEP {\bf 0902 } (2009)  044.
  [arXiv:0812.3761 [hep-th]].


\bibitem{Chaichian:1998kp}
  M.~Chaichian, A.~Demichev and P.~Presnajder,
  ``Quantum field theory on noncommutative space-times and the persistence  of
  ultraviolet divergences,''
  Nucl.\ Phys.\  B {\bf 567} (2000) 360
  [arXiv:hep-th/9812180].

\bibitem{Connes:1997cr}
  A.~Connes, M.~R.~Douglas, A.~S.~Schwarz,
  ``Noncommutative geometry and matrix theory: Compactification on tori,''
  JHEP {\bf 9802 } (1998)  003.
  [hep-th/9711162].




\bibitem{Grosse:1999ci}
  H.~Grosse, A.~Strohmaier,
  ``Towards a nonperturbative covariant regularization in 4-D quantum field theory,''
  Lett.\ Math.\ Phys.\  {\bf 48 } (1999)  163-179.
  [hep-th/9902138].

\bibitem{Balachandran:2001dd}
  A.~P.~Balachandran, B.~P.~Dolan, J.~-H.~Lee, X.~Martin, D.~O'Connor,
  ``Fuzzy complex projective spaces and their star products,''
  J.\ Geom.\ Phys.\  {\bf 43 } (2002)  184-204.
  [hep-th/0107099].

\bibitem{Grosse:2004wm}
  H.~Grosse, H.~Steinacker,
  ``Finite gauge theory on fuzzy CP**2,''
  Nucl.\ Phys.\  {\bf B707 } (2005)  145-198.
  [hep-th/0407089].


\bibitem{Chatzistavrakidis:2009ix}
  A.~Chatzistavrakidis, H.~Steinacker, G.~Zoupanos,
  ``On the fermion spectrum of spontaneously generated fuzzy extra dimensions with fluxes,''
  Fortsch.\ Phys.\  {\bf 58 } (2010)  537-552.
  [arXiv:0909.5559 [hep-th]].

\bibitem{Aoki:2010gv} 
  H.~Aoki,
  ``Chiral fermions and the standard model from the matrix model compactified
  on a torus,''
  arXiv:1011.1015 [hep-th].


\bibitem{Aoyama:2010ry}
  T.~Aoyama, J.~Nishimura and T.~Okubo,
  ``Spontaneous breaking of the rotational symmetry in dimensionally reduced
  super Yang-Mills models,''
  Prog.\ Theor.\ Phys.\  {\bf 125} (2011) 537
  [arXiv:1007.0883 [hep-th]]; 
  J.~Nishimura, F.~Sugino,
  ``Dynamical generation of four-dimensional space-time in the IIB matrix model,''
  JHEP {\bf 0205 } (2002)  001.
  [hep-th/0111102].



\bibitem{Madore:1991bw}
  J.~Madore,
  ``The Fuzzy sphere,''
  Class.\ Quant.\ Grav.\  {\bf 9 } (1992)  69-88.


\end{thebibliography}
\end{document}